\documentclass[12pt]{scrartcl}

\deffootnote{1.5em}{1em}{\textsuperscript{\thefootnotemark}}
\setkomafont{disposition}{\bfseries}
\setkomafont{title}{}

\usepackage{CJK,CJKspace,CJKpunct}
\pdfmapline{=unisong@Unicode@ <ipam.ttf}

\usepackage[T1]{fontenc}
\usepackage[pdftex]{graphicx}
\usepackage{amsmath,amssymb}
\usepackage{amsthm}

\newtheorem{theorem}{Theorem}
\newtheorem{lemma}[theorem]{Lemma}
\newtheorem{conjecture}[theorem]{Conjecture}
\theoremstyle{definition}
\newtheorem*{FencePatrollingProblem}{Interval Patrolling Problem}
\newtheorem*{CirclePatrollingProblem}{Circle Patrolling Problem}
\newtheorem*{PointPatrollingProblem}{Point Patrolling Problem}
\newtheorem*{DiscretizedPointPatrollingProblem}{Discretized Point Patrolling Problem}
\newtheorem*{DisjointCoveringSystemProblem}{Disjoint Covering System Problem}
\newtheorem*{DisjointResidueClassProblem}{Disjoint Residue Class Problem}
\newtheorem*{GeneralizedPointPatrollingProblem}{Generalized Point Patrolling Problem}
\newtheorem*{partition}{Partition-based strategy}
\newtheorem*{runners}{Runners strategy}

\newcommand{\Rset}{\mathbb R}
\newcommand{\Zset}{\mathbb Z}

\newcommand{\classNP}{\mathsf{NP}}

\begin{document}

\title{Simple strategies versus optimal schedules in multi-agent patrolling%
  \thanks{Accepted for publication in \emph{Theoretical Computer Science}~\cite{journal}.  A preliminary version of this paper was announced at the Ninth International Conference on Algorithms and Complexity (CIAC)~\cite{ciac}.}
  \thanks{This work was supported in part by the Asahi Glass Foundation, JSPS KAKENHI JP17K19960, and a joint research program of NTT and Kyushu University.}}
\author{Akitoshi Kawamura\and Makoto Soejima}
\date{}

\maketitle

\begin{abstract}
Suppose that a set of mobile agents, each with a predefined maximum speed,
want to patrol a fence together
so as to minimize the longest time interval
during which a point on the fence is left unvisited.
In 2011, Czyzowicz, G\k asieniec, Kosowski and Kranakis studied
this problem for the settings where the fence
is an interval (a line segment) and a circle,
and conjectured that
the following simple strategies are always optimal:
for Interval Patrolling,
the simple strategy
partitions the fence into subintervals,
one for each agent,
and lets each agent move back and forth in the assigned subinterval
with its maximum speed;
for Circle Patrolling,
the simple strategy is to choose a number $r$,
place the $r$ fastest agents equidistantly around the circle,
and move them at the speed of the $r$th agent.
Surprisingly, these conjectures were then proved false:
schedules were found (for some settings of maximum speeds)
that slightly outperform the simple strategies.

In this paper, we are interested in the ratio between the performances of
optimal schedules and simple strategies.
For the two problems,
we construct schedules that are
$4 / 3$ times (for Interval Patrolling) and
$21 / 20$ times (for Circle Patrolling) as good,
respectively, as the simple strategies.
We also propose a new variant,
in which we want to patrol a single point
under the constraint that
each agent can only visit the point
some predefined time after its previous visit.
We obtain some similar ratio bounds
and $\classNP$-hardness results related to this problem.
\end{abstract}

\section{Introduction}
\label{section: introduction}

In \emph{patrolling problems}, 
a set of mobile agents are deployed in order to
protect or supervise a given terrain,
and the goal is to leave no point unattended for a long period of time.
Recent studies~\cite{cgc59} have shown that
finding an optimal strategy is not at all straightforward,
even when the terrain is as simple as it could be.
We continue this line of research
in three basic settings:
patrolling a line segment, a circle, and a point.

\subsection{Interval patrolling}
\label{subsection: interval patrolling introduction}

In 2011, Czyzowicz et al.~\cite{esa2011}
proposed some simple variants of patrolling problems.
One of them was as follows\footnote{%
  In their problem, time is in $[0, \infty)$ instead of $\Rset$.
  This difference does not affect our analysis essentially,
  as we will see at the end of Section~\ref{subsection: discretization}.
}:

\begin{FencePatrollingProblem}
We want to patrol an interval (called the fence) using $k$ mobile agents.
We are given their maximum speeds $v _1$, \ldots, $v _k \geq 0$ 
and the \emph{idle time} $T > 0$.
For each point~$x$ on the fence and time $t \in \Rset$,
there must be an agent who visits $x$ during the interval $[t, t+T)$.
How long can the fence be?
\end{FencePatrollingProblem}

Note that if we can patrol a fence of length $L$ with idle time $T$,
we can patrol a fence of length $\alpha L$ with idle time $\alpha T$
by scaling the whole schedule by any $\alpha > 0$.
Thus, we are only interested in the ratio of $L$ and $T$.
Unless stated otherwise, we fix the idle time to $T = 1$.

Czyzowicz et al.~\cite{esa2011} considered the following simple strategy
that patrols a fence of length $(v _1 + \dots + v _k) / 2$ (with idle time~$1$),
and pointed out that no schedule can patrol more than twice as long a fence as this strategy:

\begin{partition}
Divide the fence into $k$ subintervals
of lengths proportional to the maximum speeds of the agents,
and let each agent move back and forth on its corresponding subinterval.
\end{partition}

They conjectured that this gives the optimal schedule,
but it was disproved later:
Kawamura and Kobayashi~\cite{isaac} found
a setting of maximum speeds $v _1$, \ldots, $v _k$
and a schedule that patrols a fence slightly longer
than the partition-based strategy does.
Thus, a natural question arises:
what is the biggest ratio between the performance of an optimal schedule and
that of the partition-based strategy?
Formally, we want to determine the smallest constant $c$ such that
no schedule can patrol a fence of length
$c (v _1 + \dots + v _k) / 2$. 

Czyzowicz et al.'s result~\cite{esa2011} mentioned above means that $1 \leq c \leq 2$,
and their conjecture was that $c = 1$.
Kawamura and Kobayashi's example shows that $c \geq 42/41$.
Later this lower bound was improved to $25/24$~\cite{Chen13, Dumitrescu2014}.
In Section~\ref{subsection: ratio 4/3},
we further improve the lower bound to $4/3$.

\subsection{Unidirectional circle patrolling}
\label{subsection: circle patrolling introduction}

Czyzowicz et al.~\cite{esa2011} also proposed the following
(Unidirectional) Circle Patrolling Problem,
which we will discuss in Section~\ref{section: circle}:

\begin{CirclePatrollingProblem}
We want to patrol a circle using $k$ mobile agents.
We are given their maximum speeds $v_1$, \ldots, $v_k \geq 0$
and the idle time~$T > 0$. 
For each point~$x$ on the circle and time $t \in \Rset$,
there must be an agent who visits the point~$x$ during the interval $[t, t+T)$.
Each agent $i$ can move along the circle clockwise
at any speed between $0$ and $v _i$.
How long can the perimeter of the circle be?
\end{CirclePatrollingProblem}

Again, we fix the idle time to $T = 1$ unless noted otherwise.

Czyzowicz et al.~\cite{esa2011} considered the following simple strategy,
which patrols a perimeter of length $\max _r r v_r$ (with idle time $1$):

\begin{runners}
Without loss of generality, we can assume that $v_1 \geq \cdots \geq v_k$.
We place the fastest $r$ agents equidistantly on the circle
and let them move at speed $v _r$,
choosing the optimal $r$.
\end{runners}

They conjectured that this is optimal,
but Dumitrescu et al.\ \cite[Theorem~1]{Dumitrescu2014} found an example
refuting this.
Again, we may ask about the ratio of this simple strategy and the optimal schedule.
For this problem, we conjecture that the Runners Strategy is
not even a constant-ratio approximation strategy.
Formally, we suspect that for any constant $c$,
there exist ($k$ and) $v_1$, \ldots, $v_k$ such that
we can patrol a perimeter of $c \max _r r v_r$.

In attempt to progress towards the conjecture,
we relate the problem
to what we call \emph{constant gap families}
in Section~\ref{subsection: constant gap families}.
Using this relation,
we construct a schedule that patrols $1.05 \max _r r v _r$
in Section~\ref{subsection: circle patrolling 1.05}.

\subsection{Point patrolling}
\label{subsection: point patrolling introduction}

Consider now a variant of the above circle patrolling problem
where there is one special point on the circle,
and our goal is to patrol (i.e., visit frequently enough) just this point,
rather than all points on the circle.
Here, let us assume that a ``visit'' to the point only happens
when an agent \emph{arrives} at this point
(so staying at the point forever does not solve the problem).
In this situation, the maximum speed of each agent simply sets
a lower bound on the time between its two consecutive visits.
This motivates us to introduce the following problem,
which we discuss in Section~\ref{section: point}:

\begin{PointPatrollingProblem}
We want to patrol a point using $k$ agents.
We are given minimum gaps $a _1$, \ldots, $a _k > 0$ and the idle time $T > 0$.
For each agent~$i = 1$, \ldots, $k$,
the gaps between two distinct visits by $i$ must always be at least $a _i$.
Can we \emph{patrol} the point with idle time $T$
so that,
for each time $t \in \Rset$,
there is an agent~$i$ that visits the point
during the time interval $[t, t + T)$?
\end{PointPatrollingProblem}

We could of course consider the problem where
we are given $(a _1, \ldots, a _k)$ and asked to minimize $T$,
but this optimization problem can be reduced to the above decision version
by binary search.

As we mention in Section~\ref{subsection: discretized point patrolling},
this decision problem can be further reduced to a discrete version
where the minimum gaps $a _1$, \ldots, $a _k$ are positive integers and
the goal is to visit the point at each integer time.
We study the relation between this problem and the quantity $1 / a _1 + \cdots + 1 / a _k$.

In Section~\ref{section: complexity},
we analyze the complexity of problems that are related to this discretized problem.

\subsection{Related work}

Problems essentially similar to patrolling are discussed under various names and in various real-world contexts,
such as
maintenance of facilities \cite{AGH98, BBNS02},
managing delivery vehicles \cite{campbell2005vehicle},
periodic scheduling \cite{BBNS02, SST09}, and
periodic latency problems \cite{charlemagne}.

While our problems impose frequent visit of every point on the fence
(an interval or a circle),
some authors consider settings where only parts of the terrain
need to be visited often~\cite{Collins13}.
Other work considers patrolling of more general graphs
(rather than just paths and cycles), in which case
we may impose frequent visit of
vertices \cite{Gorain-Mandal, pasqualetti_franchi_bullo_2010, charlemagne, noshiro_thesis}
or edges \cite{YanovskiWB2003}.
Some settings in literature are perhaps motivated by realistic consideration:
the patrollers maybe do not see just points but have visibility regions
\cite{Chen13, visibility};
they may need to stay at a point for a while or walk more slowly
when they are looking for intruders rather than just moving \cite{charlemagne, CGKMP16}.

While we are interested in minimizing idle time,
i.e., the longest time during which any point is left unvisited,
there can be other objectives, such as
minimizing the average \cite{ESK08}
and using as few agents~\cite{campbell2005vehicle}
or as little global knowledge as possible.
Some authors consider the intruder as an intelligent player
and study the situation in probabilistic or game-theoretic ways
\cite{alpern_patrolling_2011, papadaki_patrolling_2016}.
Our setting
can be viewed as a special case
where we are required to succeed in intercepting the intruder
with probability $1$.

\section{Interval patrolling}

In the Interval Patrolling Problem
(Section~\ref{subsection: interval patrolling introduction}),
the fence is an interval $[0, L]$,
and a \emph{schedule} (for agents with maximum speeds $v _1$, \ldots, $v _k$)
is formally
a $k$-tuple $(a _1, \ldots , a _k)$ of functions,
where each
$a _i \colon \Rset \to \Rset$ satisfies
$\lvert a _i (s) - a _i (t) \rvert \leq v _i \cdot \lvert s - t \rvert$
for all $s$, $t \in \Rset$.
The schedule \emph{patrols} the fence with idle time $T$ if
for each time $t \in \Rset$ and each location $x \in [0, L]$,
there are an agent~$i$ and a time $t' \in [t, t + T)$ such that $a _i (t') = x$.

In Section~\ref{subsection: ratio 4/3},
we prove that for any $c < 4 / 3$,
there exists a schedule that patrols an interval $c$ times as long as the partition-based strategy.
This improves the same claim for $c < 25 / 24$ established previously~\cite{Chen13, Dumitrescu2014}.

In Section~\ref{subsection: discretization},
we prove that any schedule can be approximated arbitrarily closely by a periodic schedule.
Thus, for any $\varepsilon > 0$, we can find in finite time
a schedule that is $1 - \varepsilon$ times as good as any schedule.

\subsection{A schedule patrolling a long interval}
\label{subsection: ratio 4/3}

The goal of this section is the following:

\begin{theorem}
\label{theorem: strategy 4/3}
For any $c < 4 / 3$,
there are maximum speeds $v _1$, \ldots, $v _k$
and a schedule that patrols an interval of length
$c (v _1 + \dots + v _k) / 2$
(with idle time~$1$).
\end{theorem}

\begin{proof}
We construct, for each pair of positive integers $n$ and $L$,
a schedule that
patrols an interval of length $L$ (with idle time~$1$)
using
$n + L - 1$ agents with maximum speed $1$ and
$n L$ agents with maximum speed $1/(2n-1)$.
If the same set of agents follows the partition-based strategy,
they would patrol the length
$\frac 1 2 (n + L - 1 + n L / (2 n - 1))$.
The ratio between $L$ and this approaches $4 / 3$
when both $n$ and $L / n$ are big,
and hence we have the theorem.
Our schedule is as follows
(Figure~\ref{figure: four_thirds}).
\begin{figure}
\begin{center}
\includegraphics{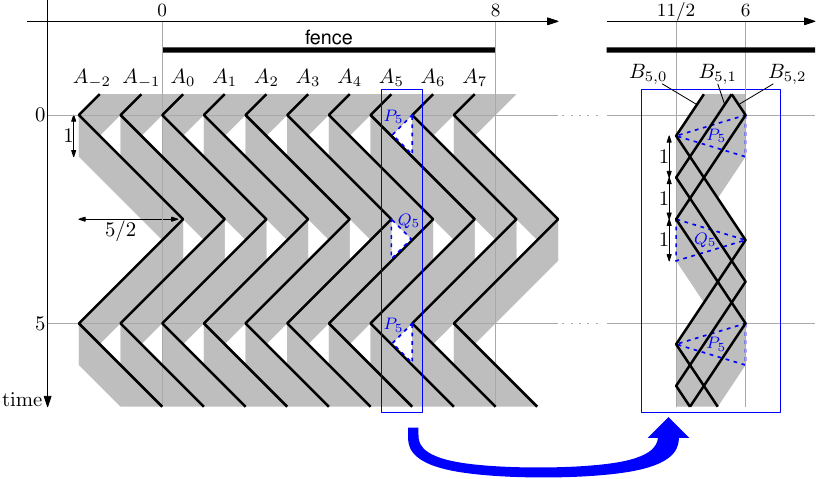}
\caption{%
The strategy in the proof of Theorem~\ref{theorem: strategy 4/3} when $n = 3$ and $L = 8$.
The trajectories of the agents are the thick solid lines,
and the regions they cover (the points that have been visited during the past unit time)
are shown shaded.
The $n + L - 1$ faster agents $A _{-n + 1}$, \ldots, $A _{L - 1}$ (left) move back and forth
with period $2 n - 1$, but leave some triangular regions (dotted) uncovered.
These regions are covered by the $n L$ slow agents $B _{0, 0}$, \ldots, $B _{n-1, L-1}$
(right; scaled up horizontally for clarity).}
\label{figure: four_thirds}
\end{center}
\end{figure}
\begin{itemize}
\item
Each of the $n + L - 1$ agents~$A _i$ ($-n < i < L$) with speed $1$
visits the locations $i$ and $i + n - 1/2$
alternately (at its maximum speed);
it is at location~$i$ at time~$0$.
(This means that some agents occasionally step out of the fence $[0, L]$;
to avoid this,
we could simply modify the schedule
so that they stay at the endpoint for a while.)
\item
Each of the $n L$ agents $B _{i, j}$ ($0 \leq i < L$, $0 \leq j < n$)
with speed $1 / (2 n - 1)$
visits the locations $i + 1/2$ and $i + 1$
alternately (at its maximum speed);
it is at location $i + 1/2$ at time $j + 1/2$.
\end{itemize}
We claim that this schedule indeed patrols the fence.
That is,
every $(x, t) \in [0, L] \times \Rset$ is \emph{covered} by some agent
in the sense that
$x$ is visited by this agent during the time interval $[t-1, t]$.

To see this,
note that every agent in this schedule repeats its movement
with a period of $2 n - 1$.
Thus, in this proof, we consider the time modulo $2 n - 1$.
Also, it is straightforward to verify that
the only regions (in the time period $[0, 2 n - 1]$)
not covered by the faster agents $A _i$ ($-n < i < L$)
are regions $P _k$ and $Q _k$, for $k = 0$, \ldots , $L - 1$, where
\begin{itemize}
\item
  $P _k$ is the triangle with vertices $(x, t) = (k + 1, 0), (k + 1, 1), (k + 1/2, 1/2)$, and
\item
  $Q _k$ is the triangle with vertices $(x, t) = (k+1/2, n-1/2), (k+1/2, n+1/2), (k+1, n)$
\end{itemize}
(Figure~\ref{figure: four_thirds}, dotted lines).
The region $P _k$ is covered by the agents $B _{k, i}$
as they move from $x = k + 1$ to $x = k + 1 / 2$
during the time interval $[i + n, i + 2 n - 1 / 2]$.
The region $Q _k$ is covered by the agents $B _{k, i}$
as they move from $x = k + 1 / 2$ to $x = k + 1$
during the time interval $[i + 1 / 2, i + n]$.
\end{proof}

We conjecture\footnote{%
  While this paper was under review,
  a disproof of this conjecture was published by Haeupler et al.~\cite{HKMPP19}.
  They constructed, for each $\varepsilon > 0$,
  an example where the agents can patrol a fence $2 (1 - \varepsilon )$ times
  as long as they would with the partition-based strategy.
}
that the above construction is optimal:

\begin{conjecture}
No schedule can patrol an interval that is more than $4 / 3$ times as long as
the partition-based strategy.
\end{conjecture}

\subsection{Zigzag schedules}
\label{subsection: discretization}

One of the difficulties about interval (or other) patrolling is that
schedules consist of real functions $a _i \colon \Rset \to \Rset$
and we hence cannot check all possible schedules exhaustively.
Below, we prove that,
for the purpose of discussing bounds on the length of the fence
(as we did in Section~\ref{subsection: ratio 4/3}),
we may restrict attention to a certain class of periodic schedules.
Using this, we show that,
given the maximum speeds $v _1$, \ldots , $v _k > 0$ of the agents
together with a positive number $\varepsilon > 0$,
we can find a schedule for them to patrol
an interval of length at least $1 - \varepsilon $ times
the maximum length that can be patrolled.

The movement of an agent during a time interval $
[t _{\mathrm{start}}, t _{\mathrm{end}}] \subseteq \Rset
$ is
represented by a function $
a \colon [t _{\mathrm{start}}, t _{\mathrm{end}}] \to \Rset
$. This function is called a \emph{$(v, \xi )$-zigzag movement}, for $v$, $\xi > 0$
(Figure~\ref{figure: zigzag}),
\begin{figure}
\begin{center}
\includegraphics{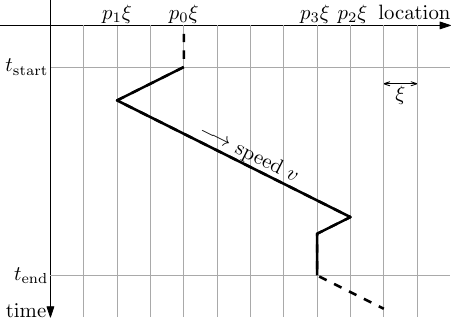}
\caption{%
A $(v, \xi)$-zigzag movement on a time interval~$[t _{\mathrm{start}}, t _{\mathrm{end}}]$.
}
\label{figure: zigzag}
\end{center}
\end{figure}
if there are integers $p _0$, $p _1$, $p _2$, $p _3 \in \Zset$
such that the agent
\begin{itemize}
\item
starts at time $t _{\mathrm{start}}$ at location $p _0 \xi $,
\item
moves at speed~$v$
until it reaches $p _1 \xi $,
\item
moves at speed~$v$
until it reaches $p _2 \xi $,
\item
moves at speed~$v$
until it reaches $p _3 \xi $,
\item
and then stays there until time $t _{\mathrm{end}}$.
\end{itemize}
For this movement to be possible,
the entire route must be short enough to be traveled with speed $v$; that is,
\begin{equation}
\label{equation: zigzag bound}
\lvert p _0 - p _1 \rvert \xi + \lvert p _1 - p _2 \rvert \xi + \lvert p _2 - p _3 \rvert \xi \leq \tau v,
\end{equation}
where $\tau := t _{\mathrm{end}} - t _{\mathrm{start}}$ is the length of the
time interval.

The following lemma says that
any movement of an agent on a time interval can be turned into
a zigzag movement taking an only slightly longer time
without shrinking the set of points traveled.

\begin{lemma}
\label{lemma: zigzag}
For any positive constants $\delta $, $\tau $, $v > 0$,
we have the following for all sufficiently small $\xi > 0$.
For any function $
a \colon [t _{\mathrm{start}}, t _{\mathrm{end}}] \to \Rset
$ on an interval of length $\tau$
such that $\lvert a (s) - a (t) \rvert \leq v \cdot \lvert s - t \rvert$
for all $s$, $t \in [t _{\mathrm{start}}, t _{\mathrm{end}}]$,
there is a $(v, \xi)$-zigzag movement $
a' \colon
[(1 + \delta ) t _{\mathrm{start}}, (1 + \delta ) t _{\mathrm{end}}] \to \Rset
$
such that
\begin{itemize}
\item
$\displaystyle
a' \bigl( (1 + \delta) t _{\mathrm{start}} \bigr) = \biggl\lfloor \frac{a (t _{\mathrm{start}})}{\xi} \biggr\rfloor \xi
$ and
$\displaystyle
a' \bigl( (1 + \delta) t _{\mathrm{end}} \bigr) = \biggl\lfloor \frac{a (t _{\mathrm{end}})}{\xi} \biggr\rfloor \xi$;
\item
any location visited by $a$ is visited by $a'$---that is,
for each $t \in [t _{\mathrm{start}}, t _{\mathrm{end}}]$
there is $t' \in [(1 + \delta ) t _{\mathrm{start}}, (1 + \delta ) t _{\mathrm{end}}]$ such that
$a' (t') = a (t)$.
\end{itemize}
\end{lemma}

\begin{proof}
Let $\xi \leq (t _{\mathrm{end}} - t _{\mathrm{start}}) v \delta / 5$.
Suppose that $a$ takes its minimum and maximum at
$t _{\mathrm{min}}$, $t _{\mathrm{max}} \in [t _{\mathrm{start}}, t _{\mathrm{end}}]$, respectively.
We may assume that $t _{\mathrm{min}} \leq t _{\mathrm{max}}$
(the other case can be treated similarly).
Define $a'$ to be the $(v, \xi )$-zigzag movement specified by
\begin{align*}
p _0 & = \biggl\lfloor \frac{a (t _{\mathrm{start}})}{\xi} \biggr\rfloor, &
p _1 & = \biggl\lfloor \frac{a (t _{\mathrm{min}})}{\xi} \biggr\rfloor, &
p _2 & = \biggl\lceil \frac{a (t _{\mathrm{max}})}{\xi} \biggr\rceil, &
p _3 & = \biggl\lfloor \frac{a (t _{\mathrm{end}})}{\xi} \biggr\rfloor
\end{align*}
(see the beginning of Section~\ref{subsection: discretization} for the meaning of these numbers).
This is indeed possible:
we have \eqref{equation: zigzag bound} for $
\tau = (1 + \delta ) (t _{\mathrm{end}} - t _{\mathrm{start}})
$ because
\begin{align*}
&
  (p _{0} - p _{1}) \xi + (p _{2} - p _{1}) \xi + (p _{2} - p _{3}) \xi
\notag
\\
&
\leq
  \bigl( a (t _{\mathrm{start}}) - a (t _{\mathrm{min}}) + \xi \bigr)
 +
  \bigl( a (t _{\mathrm{max}}) - a (t _{\mathrm{min}}) + 2 \xi \bigr)
 +
  \bigl ( a (t _{\mathrm{max}}) - a (t _{\mathrm{end}}) + 2 \xi \bigr)
\notag
\\
&
=
  \bigl( a (t _{\mathrm{start}}) - a (t _{\mathrm{min}}) \bigr)
 +
  \bigl( a (t _{\mathrm{max}}) - a (t _{\mathrm{min}}) \bigr)
 +
  \bigl( a (t _{\mathrm{max}}) - a (t _{\mathrm{end}}) \bigr)
 +
  5 \xi
\notag
\\
&
\leq
  (t _{\mathrm{min}} - t _{\mathrm{start}}) v + (t _{\mathrm{max}} - t _{\mathrm{min}}) v + (t _{\mathrm{end}} - t _{\mathrm{max}}) v + 5 \xi
\notag
\\
&
=
 (t _{\mathrm{end}} - t _{\mathrm{start}}) v + 5 \xi
\leq
 (1 + \delta) (t _{\mathrm{end}} - t _{\mathrm{start}}) v.
\end{align*}
It is straightforward to see
that this zigzag movement has the claimed properties.
\end{proof}

Next, we prove in the following lemma that
any schedule can be converted into one
that consists of zigzag movements
without deteriorating the idle time too much.

For positive $\xi$, $\tau >0$,
a schedule $(a _1, \ldots, a _k)$
(for $k$ agents with maximum speeds $v _1$, \ldots , $v _k$)
is called
a \emph{$(\xi , \tau )$-zigzag schedule} if
the movement of each agent $i = 1$, \ldots, $k$
during each time interval $[m \tau, (m + 1) \tau]$, $m \in \Zset$,
is a $(v _i, \xi)$-zigzag movement.

\begin{lemma}
\label{lemma: zigzag schedule}
For any $\varepsilon > 0$ and speeds $v _1$, \ldots , $v _k > 0$,
there are $\xi > 0$ and $\tau' > 0$ satisfying the following.
Suppose that there is a schedule for a set of agents
with maximum speeds $v _1$, \ldots, $v _k$
that patrols a fence with some idle time $T > 0$.
Then there is a $(\xi, \tau')$-zigzag schedule
for the same set of agents
that patrols the same fence
with idle time $T (1 + \varepsilon)$.
\end{lemma}

\begin{proof}
We show that it suffices to
let $\xi $ be so small that we have the claim of
Lemma~\ref{lemma: zigzag} for
\begin{align*}
\delta & = \frac \varepsilon 2, &
\tau & = \frac{T \delta }{2 (1 + \delta)}
\end{align*}
and for all speeds~$v = v _i$, and to let
$\tau' = (1 + \delta ) \tau$.

Using the schedule $(a _1, \ldots , a _k)$ that we start with,
we define the claimed $(\xi, \tau')$-zigzag schedule
$(a' _1, \ldots , a' _k)$ as follows.
For each agent~$i$ and each $m \in \Zset$,
we define $a' _i$ on the time interval $[m \tau', (m + 1) \tau']$
to be the zigzag movement obtained by Lemma~\ref{lemma: zigzag}
from the movement $a _i$ during $[m \tau, (m + 1) \tau]$.
This defines $a _i$ consistently (at multiples of $\tau$) because of the
first property in Lemma~\ref{lemma: zigzag}.

To see that this schedule $(a' _1, \ldots, a' _k)$ patrols the fence
as claimed,
suppose that a location on the fence is left unvisited
by the schedule $(a' _1, \ldots, a' _k)$
during a time interval $[\underline t, \overline t]$
of length $T (1 + \varepsilon)$,
and hence during its subinterval $
[
 \lceil \underline t / \tau' \rceil \tau',
 \lfloor \overline t / \tau' \rfloor \tau'
]
$.
By the second property in Lemma~\ref{lemma: zigzag},
this point is also left unvisited
by the schedule $(a _1, \ldots, a _k)$
during the time interval $
[
 \lceil \underline t / \tau' \rceil \tau,
 \lfloor \overline t / \tau' \rfloor \tau
]
$, whose length is
\begin{equation*}
 \biggl(
   \biggl\lfloor \frac{\overline t}{\tau'} \biggr\rfloor
  -
   \biggl\lceil \frac{\underline t}{\tau'} \biggr\rceil
 \biggr) \tau
\geq
 \biggl(
  \frac{\overline t}{\tau'} - \frac{\underline t}{\tau'} - 2
 \biggr)
 \tau
=
 \frac{T (1 + \varepsilon)}{\tau'} \tau - 2 \tau
=
 \frac{T (1 + 2 \delta)}{1 + \delta} - 2 \tau
=
 T.
\qedhere
\end{equation*}
\end{proof}

The next lemma says that a zigzag schedule can be made periodic
without changing the idle time.

\begin{lemma}
\label{lemma: periodic}
Suppose that there is a $(\xi, \tau )$-zigzag schedule
for a set of agents
that patrols a fence with some idle time.
Then
there is a periodic $(\xi, \tau)$-zigzag schedule for the same agents
that patrols the same fence with the same idle time.
\end{lemma}

\begin{proof}
Let $[a, b]$ be the fence,
$k$ be the number of agents,
and $T$ be the idle time.
We may assume that in the given ($(\xi, \tau)$-zigzag) schedule,
every agent stays within
$[A \xi , B \xi]$, where
$A = \lfloor a / \xi \rfloor$ and $B = \lceil b / \xi \rceil$,
i.e., it never goes far off the fence.
In such a schedule,
the movement of each agent
during each time interval $[m \tau , (m + 1) \tau ]$,
$m \in \Zset$,
is specified by a quadruple of integers
$p _0$, $p _1$, $p _2$, $p _3 \in \{A, A + 1, \ldots , B\}$,
and hence there are at most $(B - A + 1) ^4$ possible such movement.

Let $Q = \lceil T / \tau \rceil $.
For each $m \in \Zset$,
there are at most $(B - A + 1) ^{4 k Q}$ possible
ways that the $k$ agents can move during the time interval
$[m \tau , (m + Q) \tau ]$.
Since this is finite, there are integers $m _0$, $m _1$ with $m _0 < m _1$ such that
in the given $(\xi , \tau)$-zigzag schedule,
all agents move during the time interval
$[m _1 \tau , (m _1 + Q) \tau]$
in exactly the same way as they did during
$[m _0 \tau , (m _0 + Q) \tau]$.
Consider the periodic schedule, with period $(m _1 - m _0) \tau$,
where each agent perpetually repeats
its movement during $[m _0 \tau , m _1 \tau ]$ in the original schedule.
This schedule patrols the fence with idle time $T$,
because the movement of the agents during any time period of length $T$
is identical to their movement in the original schedule
in a length-$T$ subinterval of $[m _0 \tau, (m _1 + Q) \tau]$.
\end{proof}

Using the above lemmas,
we obtain an algorithm that solves the Interval Patrolling Problem
with arbitrarily high precision
in the following sense.
Suppose that there is a schedule that patrols a fence of length $L > 0$ with idle time $T > 0$
using agents with maximum speeds $v _1$, \ldots, $v _k$.
By Lemma~\ref{lemma: zigzag schedule},
there is a $(\xi, \tau)$-zigzag schedule
that patrols a fence of length $(1 - \varepsilon ) L$,
for some $\xi$, $\tau > 0$ determined by the inputs $\varepsilon$ and $v _1$, \ldots, $v _k$.
By Lemma~\ref{lemma: periodic},
there is a $(\xi , \tau)$-zigzag schedule with period~$p$
that patrols the same length $(1 - \varepsilon) L$,
for some $p > 0$ determined by the inputs.
Since there are only finitely many $(\xi, \tau)$-zigzag schedules with period~$p$,
we can check all of them in a finite (though long) time.
Thus,

\begin{theorem}
There is an algorithm that,
given $v _1$, \ldots, $v _k$, $T$ and $\varepsilon > 0$,
finds a schedule that patrols a fence of length at least $1 - \varepsilon $ times
the length of the fence patrolled by the same agents using any schedule.
\end{theorem}

In previous work~\cite{esa2011, isaac},
a schedule was defined as functions on
the half\textcompwordmark line $[0, +\infty)$ (instead of $\Rset$)
and the requirement for patrolling was that
each location be visited in every length-$T$
time interval contained in this half\textcompwordmark line.
Note that the argument for Lemmas \ref{lemma: zigzag schedule} and \ref{lemma: periodic} in this section stays valid
when we start with a patrolling schedule on $[0, +\infty)$.
In particular,
a patrolling schedule on $[0, +\infty)$ can be
converted to a (periodic) schedule on $\Rset$
without essentially worsening the idle time.
Therefore, our slight deviation in the definition does not affect the ratio bounds
(discussed in Section~\ref{subsection: ratio 4/3}).

\section{Circle patrolling}
\label{section: circle}

As mentioned in Section~\ref{subsection: circle patrolling introduction},
we conjecture that the Runners Strategy for the Circle Patrolling Problem
does not have a constant approximation ratio, i.e.,
that for an arbitrarily large $c$,
we can find a set of agents with maximum speeds $v _1$, \ldots, $v _k$
and their schedule that patrols a perimeter of length $c \max _r r v _r$.
By scaling, we can assume that $\max _r r v _r = 1$.
In this case, $v _i \leq 1 / i$ for each $i = 1$, \ldots , $k$.
Thus, we may and will henceforth assume that $v _i = 1 / i$.

\subsection{Constant gap families}
\label{subsection: constant gap families}

In this section, we argue that the above conjecture is
equivalent to a perhaps simpler statement about
what we call \emph{constant gap families} below.

Suppose that we want to patrol (with idle time $1$)
a circle with perimeter $c > 1$.
By definition, a successful patrolling schedule is one in which,
for each time $t \in \Rset$ and each point $x$ on the circle,
some agent visits $x$ during the time interval $[t, t + 1]$.
Now, instead of imposing this for \emph{every} $x$,
consider the same condition for (each $t$ and)
the specific point $x = ct \bmod c$:
we say that an agent \emph{covers} $t$ if it visits $ct \bmod c$
during the time interval $[t, t + 1]$.
In other words, we consider an imaginary particle
that moves along the circle at speed~$c$,
and say that $t$ is covered if
the point that the particle passes at time $t$
is then visited by an agent before the particle comes again
(Figure~\ref{figure: circle}).
\begin{figure}
\begin{center}
\includegraphics{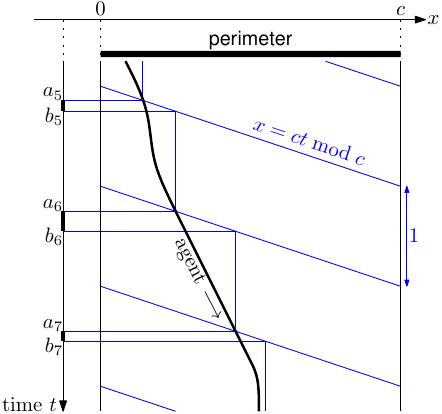}
\caption{%
  Time intervals $[a _j, b _j]$ that are covered by an agent.
  If the maximum speed of the agent is $v$ ($\leq 1$),
  the time $a _{j + 1} - a _j$ that it takes for it to fall one lap behind the particle (of speed~$c$)
  is at most $c / (c - v)$.
  Thus, $b _j - a _j = (a _{j + 1} - 1) - a _j \leq v / (c - v)$.
}
\label{figure: circle}
\end{center}
\end{figure}

For a schedule to patrol the circle,
every time $t \in \Rset$ must be covered by some agent.
This is a necessary, but not a sufficient, condition,
because for each point~$x$ on the circle,
we are requiring a visit only for some of the length-$1$ time intervals,
namely those that start and end when the particle passes $x$.
On the other hand,
this is a sufficient condition for the schedule
to patrol the circle with idle time $2$,
because every time interval of length~$2$
contains an interval of length~$1$ aligned with the particle's visits.
Thus,

\begin{lemma}
\label{lemma: circle covering}
Consider a schedule (of several agents) on a circle of perimeter $c > 1$.
\begin{enumerate}
\item
If the schedule patrols the circle with idle time $1$,
every time $t \in \Rset$ is covered by some agent.
\item
If every time $t \in \Rset$ is covered by some agent,
the schedule patrols the circle with idle time $2$.
\end{enumerate}
\end{lemma}

As shown in Figure~\ref{figure: circle},
the set of $t \in \Rset$ covered by an agent of maximum speed $v < c$
is a union of disjoint intervals $\bigcup _{j \in \Zset} [a _j, b _j]$
such that
\begin{align}
\label{equation: covered time in circle patrolling}
b _j - a _j & \leq \frac{v}{c-v},
&
a _{j + 1} - b _j & = 1
\end{align}
for all $j \in \Zset$.
Conversely,
for any family of intervals $([a _j, b _j]) _{j \in \Zset}$ satisfying
\eqref{equation: covered time in circle patrolling},
an agent with maximum speed~$v$
can cover all times in $\bigcup _{j \in \Zset} [a _j, b _j]$
by a movement in which,
for each $j \in \Zset$,
it is at point $c a _j \bmod c$ at time $a _j$.

Now, let us assume that
the maximum speed of agent~$i$ is $1 / i$,
as discussed at the beginning of Section~\ref{section: circle}.
Then the set $S _i \subseteq \Rset$ of times covered by agent $i$
is a union of intervals
satisfying \eqref{equation: covered time in circle patrolling} for $v = 1 / i$.
That is, it can be written as $
S _i = \bigcup _{j \in \Zset} [a_{i,j}, b_{i,j}]
$ with
\begin{align}
\label{equation: covered time in circle patrolling h}
b _{i, j} - a _{i, j} & \leq \frac{1}{c i - 1},
&
a _{i, j + 1} - b _{i, j} & = 1
\end{align}
for all $j \in \Zset$.
For a real number $c > 1$ and a positive integer $k$,
a $(c, k)$-\emph{constant gap family}
(henceforth a $(c, k)$-\emph{family})
is a $k$-tuple of sets $S _1, \ldots, S _k \subseteq \Rset$
satisfying $S _1 \cup \cdots \cup S _k = \Rset$
and \eqref{equation: covered time in circle patrolling h}.
Applying Lemma~\ref{lemma: circle covering}, we have:

\begin{lemma}
\label{lemma: circle}
Let $c > 1$.
\begin{enumerate}
\item
If $k$ agents with maximum speeds $1$, $1/2$, \ldots, $1/k$ can patrol
a circle of perimeter $c$ (with idle time $1$),
then there is a $(c, k)$-family.
\item
If there is a $(c, k)$-family,
then $k$ agents with maximum speeds $1$, $1/2$, \ldots , $1/k$ can patrol
a circle of perimeter $c$ with idle time $2$
(and thus can patrol a circle of perimeter $c / 2$ with idle time $1$).
\end{enumerate}
\end{lemma}

\subsection{A schedule patrolling a large circle}
\label{subsection: circle patrolling 1.05}

As we did for the partition-based strategy in interval patrolling,
we may want to ask how good the runners strategy is
in comparison to optimal schedules.
Our conjecture\footnote{%
  While this paper was under review,
  this conjecture was solved affirmatively~\cite{HKMPP19}.
} is the following:

\begin{conjecture}
\label{conjecture: circle unbounded}
The runners strategy does not have a constant approximation ratio.
That is, for any constant $c$,
there exist $v_1, \ldots , v_k$
and a schedule
that patrols a circle with perimeter $c \max _r rv_r$.
\end{conjecture}

By Lemma~\ref{lemma: circle},
this holds if and only if
for each constant $c$,
there are $k$ and a $(c, k)$-family.
Note that this is non-trivial already for $c > 1$.

We used a computer program and found a
$(2.1, 122)$-family $(S _1, \dots, S _{122})$,
which we posted in the ancillary files
of the arXiv version of this paper.
We were able to finitely describe the whole family,
because for each $i \in \{1, \ldots, 122\}$,
the set $S _i$ has a period of $500$
(i.e., for each time $t \in \Rset$, we have $t \in S _i$ if and only if $t + 500 \in S _i$),
and the endpoints of all intervals of $S _i$ are multiples of $1 / 400$.
This also implies that we can check the validity (i.e., $S _1 \cup \cdots \cup S _k = \Rset$ and the condition~\eqref{equation: covered time in circle patrolling h})
of this example in a straightforward way using a computer program.
Thus, by Lemma~\ref{lemma: circle},

\begin{theorem}
\label{theorem: 1.05}
There exist $v_1, \ldots , v_k$ and a schedule
that patrols a circle with perimeter $1.05 \max_r rv_r$.
\end{theorem}

\section{Point patrolling}
\label{section: point}

We discuss the Point Patrolling Problem
introduced in Section~\ref{subsection: point patrolling introduction}.

\subsection{Point patrolling in discrete time}
\label{subsection: discretized point patrolling}

We start by observing that this problem
can be reduced to a discretized version:

\begin{DiscretizedPointPatrollingProblem}
There are $k$ agents.
We are given positive integers $m _1$, \ldots, $m _k$.
The gap between two distinct visits by agent~$i$ must be at least $m _i$.
Determine whether there is a schedule of visits in which,
at each integer time,
at least one agent makes a visit.
\end{DiscretizedPointPatrollingProblem}

The list of integers $m = (m _1, \ldots, m _k)$ is called \emph{good}
if it admits a patrolling schedule for this problem,
and otherwise \emph{bad}.

\begin{theorem}
\label{theorem: discretizing point patrolling}
Agents with minimum gaps $a _1$, \ldots, $a _k > 0$
can patrol the point with idle time $T$
in the (non-discretized) Point Patrolling Problem
(see Section~\ref{subsection: point patrolling introduction})
if and only if $(\lceil a_1 / T \rceil, \ldots, \lceil a_k / T \rceil)$ is good.
\end{theorem}

\begin{proof}
For the `if' direction,
suppose that $(\lceil a _1 / T \rceil, \ldots, \lceil a _k / T \rceil)$ is good.
That is, there is a schedule where at each time $j \in \Zset$,
the point is visited by an agent $i _j \in \{1, \ldots , k\}$.
Consider the schedule (for the non-discretized problem) where
the corresponding visit is made by the same agent $i _j$ but
now at time $j T$.
This schedule is feasible,
because the gap between two distinct visits by agent~$i$ is now
at least $\lceil a _i / T \rceil \cdot T \geq a _i$.

Conversely, suppose that the agents achieve the idle time of $T$
in the non-discretized problem.
Let $(i _j, t _j) _{j \in \Zset} \in (\{1, \ldots, k\} \times \Rset) ^\Zset$ be
the (ordered) list of all visits in this schedule.
That is, the $j$th visit is made by agent $i _j$ at time $t _j$, so that
\begin{itemize}
\item
  $0 \leq t _{j + 1} - t _j \leq T$ for all $j \in \Zset$; and
\item
  for all $j$, $j' \in \Zset$,
  if $i _{j} = i _{j'}= i$ and $j < j'$, then $t _{j'} - t _{j} \geq a _i$.
\end{itemize}
Now consider the schedule
in which the $j$th visit is made by the same agent $i _j$ but
now at time $j$.
This is a feasible schedule for the discretized problem, because,
if agent~$i$ makes the $j$th and then the $j'$th visit,
then $j' - j \geq (t _{j'} - t _{j}) / T \geq a _i / T$
by the two conditions above.
\end{proof}

Thus, we will henceforth be interested in the
Discrete Point Patrolling Problem.
We note that this problem can be solved in finite time
as follows.
Consider a directed graph with $\prod _{i = 1} ^k m _i$ vertices,
each of which corresponds to a $k$-tuple $(b _1, \ldots, b _k)$ with
$b _i \in \{0, 1, \dots, m _i - 1\}$ for each $i$.
This vertex means that
currently agent~$i$ still needs to wait for $b _i$ time
before its next visit
(i.e., time $m _i - b _i$ has elapsed since its most recent visit).
We add an edge from a vertex $(b _1, \ldots, b _k)$
to a vertex $(b' _1, \ldots , b' _k)$
when there is $i$ such that
$(b _i, b' _i) = (0, m _i - 1)$ and
$b' _j = \max \{0, b _j - 1\}$ for each $j \neq i$.
A successful patrolling schedule corresponds to an infinite walk in this graph.
Thus, $(m _1, \ldots , m _k)$ is good if and only if this graph has a cycle.

Using standard algorithms for strongly connected components, we can solve
the problem in $O (k \prod _{i=1}^{k} m _i)$ time.
We conjecture that this problem is not solvable in polynomial time in general.
As we will see in Section~\ref{section: complexity},
similar problems turn out to be $\classNP$-complete.

In fact, it is not clear whether this problem is in $\classNP$.
The above discussion of the directed graph
shows that if we can patrol a point with the given agents,
then we can do so by a periodic schedule.
But it only gives an exponential bound $\prod _{i = 1} ^k m _i$ on this period,
which is not polynomial in the input length
(even if each number $m _i$ is given in unary).

\subsection{Quantitative conditions for point patrolling}
\label{subsection: suminv}

\newcommand{\SumInv}{\mathit{SumInv}}

For a list $m = (m _1, \ldots, m _k)$ of positive integers,
define $\SumInv (m) = 1 / m _1 + \dots + 1 / m _k$. 
There is an obvious necessary condition for $a$ to be good:

\begin{theorem}
\label{theorem: bad}
If a list $m = (m _1, \dots, m _k)$ of positive integers is good,
then $\SumInv (m) \geq 1$. 
\end{theorem}

\begin{proof}
As discussed at the end of Section~\ref{subsection: discretized point patrolling},
there is a periodic patrolling schedule.
During each period, each agent $i$ can visit
at most $1 / m _i$ fraction of the time,
so the sum of these fractions over all agents must be $\geq 1$.
\end{proof}

The converse does not hold in general (for example, $(2, 3, 5)$ is bad, see Theorem~\ref{theorem: bad example}),
but it does in a special case:

\begin{lemma}
\label{lemma: SumInv sufficient}
If a list $m = (m _1, \dots, m _k)$ of positive integers satisfies $\SumInv (m) \geq 1$,
and $m _{i - 1}$ divides $m _i$ for all $i = 2$, \dots, $k$,
then $m$ is good.
\end{lemma}

\begin{proof}
We prove by induction on $i = 1$, \ldots, $k$ that
there is a set $R _i \subseteq \{1, \ldots , m _i\}$
of size $\lvert R _i \rvert = m _i \cdot \min \{\SumInv (m _1, \dots, m _i), 1\}$
and a schedule for the first $i$ agents $1$, \ldots, $i$
that visits all times in
$m _i \Zset + R _i := \{\, m _i n + r : n \in \Zset, \ r \in R _i \,\}$.
The lemma then follows from this claim for $i = k$.

The claim is trivial for $i = 1$ (with $R _1$ being any set of size $1$).
Suppose $i > 1$.
By the assumption, $m _i$ is a multiple of $m _{i - 1}$,
and the induction hypothesis
gives a set $R _{i - 1} \subseteq \{1, \ldots, m _{i - 1}\}$
of size $m _{i - 1} \cdot \min \{\SumInv (1, \ldots , m _{i - 1}), 1 \}$
such that
the first $i - 1$ agents visit all times in $m _{i - 1} \Zset + R _{i - 1}$.
Let $R _i$ be
the set of numbers in $\{1, \ldots , m _i\}$
congruent modulo $m _{i - 1}$ to some element of $R _{i - 1}$,
so that $\lvert R _i \rvert = (m _i / m _{i - 1}) \cdot \lvert R _{i - 1} \rvert$.
If $R _{i - 1} = \{1, \dots , m _{i - 1}\}$ (and thus $R _i = \{1, \dots , m _i\}$ and $m _{i - 1} \Zset + R _{i - 1} = m _i \Zset + R _i = \Zset$),
we are done,
without even using the additional agent~$i$.
Otherwise,
add any one element $r ^* \in \{1, \ldots, m _{i - 1}\} \setminus R _{i - 1}$
to $R _i$,
and let agent~$i$ visit the times in $m _i \Zset + \{r ^*\}$.
Now $\lvert R _i \rvert = m _i \cdot \SumInv (m _1, \ldots, m _{i - 1}) + 1 = m _i \cdot \SumInv (m _1, \ldots , m _i)$.
\end{proof}

A corollary to this lemma is the following:

\begin{theorem}
\label{theorem: good}
If a list $m = (m _1, \dots, m _k)$ of positive integers satisfies $\SumInv (m) \geq 2$,
then $m$ is good.
\end{theorem}

\begin{proof}
Since the ordering of the agents does not matter,
we may assume that $m _1 \leq \dots \leq m _k$.
For each $i = 1$, \ldots, $k$,
let $e _i$
be the smallest integer with $2 ^{e _i} > m _i$.
Since $2 ^{e _i} \leq 2 m _i$,
we have $
 \SumInv (2 ^{e _1}, \ldots, 2 ^{e _k})
\geq
 1 / (2 m _1) + \dots + 1 / (2 m _k)
\geq
 1
$, and thus we can apply Lemma~\ref{lemma: SumInv sufficient}
to conclude that
$(2 ^{e _1}, \ldots , 2 ^{e _k})$ is good,
and hence so is $m$.
\end{proof}

In fact, the same argument proves something slightly stronger:
\emph{if $m _1$, \ldots, $m _k$ are positive real numbers with
$1 / m _1 + \dots + 1 / m _k \geq 2$,
then $m' = (\lfloor m _1 \rfloor + 1, \ldots, \lfloor m _k \rfloor + 1)$ is good}.
This is obvious when $m _i < 1$ for some $i = 1$, \ldots, $k$.
Otherwise, each $e _i$ satisfies $2 ^{e _i} \geq \lfloor m _i \rfloor + 1$,
and thus the conclusion holds for $m'$ instead of $m$.

This (together with Theorem~\ref{theorem: bad})
gives a 2-approximation algorithm for the
(non-discretized) Point Patrolling Problem:
Given the minimum gaps $a _1$, \ldots, $a _k > 0$ of agents,
let $T$ be the number satisfying $T / a _1 + \dots + T / a _k = 1$.
Then $(\lfloor a _1 / (2 T) \rfloor + 1, \ldots, \allowbreak \lfloor a _k / (2 T) \rfloor + 1)$ is good
by the aforementioned stronger version of Theorem~\ref{theorem: good} (with $m _i = a _i / (2 T)$),
and hence, by one direction of
Theorem~\ref{theorem: discretizing point patrolling},
we obtain a patrolling schedule with idle time $2T$.
To see that this is indeed a $2$-approximation,
i.e., that idle time $T' < T$ is unachievable,
we invoke the other direction of Theorem~\ref{theorem: discretizing point patrolling}
with $(\lceil a _1 / T' \rceil, \ldots, \allowbreak \lceil a _k / T' \rceil)$,
which is bad by Theorem~\ref{theorem: bad} because $
 \SumInv (\lceil a _1 / T' \rceil, \ldots, \allowbreak \lceil a _k / T' \rceil)
<
 T / a _1 + \dots + T / a _k
=
 1
$.

For the rest of this section,
we will be interested in improving
the constant $2$ in Theorem~\ref{theorem: good}.
The following lemma shows that
from a bad list~$m$,
we can obtain another bad list~$m'$
whose entries are smaller
but whose $\SumInv$ is not much smaller:

\begin{lemma}
\label{lemma: half}
Let $Q > 0$ and $t \geq 0$ be integers.
For a bad list $m$ of positive integers $\leq 2 ^t Q$,
there exists a bad list $m'$ of positive integers $\leq Q$
such that
\begin{equation*}
  \SumInv (m') + 1 \geq \biggl( 1 - \frac{2 (1 - 2 ^{-t})}{Q + 2} \biggr) \cdot \bigl( \SumInv (m) + 1 \bigr).
\end{equation*}
\end{lemma}

\begin{proof}
Once the lemma is proved for $t = 1$,
we can apply this special statement $t$ times with
$Q$ replaced by $2 ^i Q$ for $i = 0$, $1$, \ldots, $t - 1$
to obtain the general statement,
because
\begin{equation*}
 1 - \frac{2 (1 - 2 ^{-t})}{Q + 2}
=
 \prod _{i = 0} ^{t - 1} \biggl( 1 - \frac{2 (1 - 2 ^{-1})}{2 ^i Q + 2} \biggr)
\end{equation*}
by a straightforward induction on $t$.

Thus, we shall henceforth assume that $t = 1$.
Let $r$ and $s$ be the numbers of elements of $m$ that are $\leq Q$ and $> Q$,
respectively.
Without loss of generality,
we can assume that $m$ is in nondecreasing order, so that
\begin{equation*}
  m = (e _1, \ldots, e _r, a _1, a _2, \ldots, a _s),
\end{equation*}
where $e _1 \leq \dots \leq e _r \leq Q < a _1 \leq a _2 \leq \dots \leq a _s$.
First, define
\begin{equation*}
  m'' = (e _1, \ldots, e _r, a _2, a _2, a _4, a _4, \ldots , a _{s - p / 2}, a _{s - p / 2}),
\end{equation*}
where $p \in \{0, 1\}$ is the parity of $s$.
Thus, $m''$ is constructed from $m$ by
replacing every other agent with minimum gap $> Q$
by one with the same minimum gap as the next agent,
and possibly removing the last agent.
Since $m$ is bad, so is $m''$.
The inverse sum of $m''$ can be smaller than that of $m$,
but only slightly:
\begin{align*}
  \SumInv (m) - \SumInv (m'')
&
 =
  \sum _{j = 1} ^{(s - p) / 2} \biggl( \frac{1}{a _{2 j - 1}} - \frac{1}{a _{2 j}} \biggr) + p \cdot \frac{1}{a _s}
\\
&
 \leq
  \sum _{i = 2} ^s \biggl( \frac{1}{a _{i - 1}} - \frac{1}{a _i} \biggr) + \frac{1}{a _s}
 =
  \frac{1}{a _1} 
 \leq
  \frac{1}{Q + 1}. 
\end{align*}
Next, define
\begin{equation*}
  m' = (e _1, \ldots, e _r, b _1, \ldots, b _{(s - p) / 2}),
\end{equation*}
where $
b _j = \lceil a _{2 j} / 2 \rceil
$ for $j = 1$, \ldots, $(s - p) / 2$.
Two agents with minimum gap $a _{2 j}$ can work as
a single agent with minimum gap $b _j$,
so $m'$ is also bad.
Since this rounding $\lceil \mathord\cdot \rceil$ increases each element by a factor of $\leq \frac{Q+2}{Q+1}$,
we have
\begin{equation*}
  \SumInv (m')
 \geq
  \frac{Q+1}{Q+2} \cdot \SumInv (m'')
 \geq
  \frac{Q+1}{Q+2} \biggl( \SumInv (m) - \frac{1}{Q + 1} \biggr), 
\end{equation*}
whence $\SumInv (m') + 1 \geq \frac{Q + 1}{Q + 2} \cdot (\SumInv (m) + 1)$, as was desired.
\end{proof}

\begin{theorem}
\label{theorem: 1.546}
If a list $m$ of positive integers satisfies
$\SumInv (m) > 1.546$,
then $m$ is good.
\end{theorem}

\begin{proof}
Suppose that there is a bad $m$ with $\SumInv (m) > 1.546$.
Using Lemma~\ref{lemma: half} for $Q = 12$ (and a sufficiently large $t$),
we obtain a bad list $m'$ of integers $\leq 12$
such that $\SumInv (m') + 1 > (1 - 2 / 14) \cdot (1.546 + 1)$,
whence $\SumInv (m') > 1.1822$.
By sorting and truncating $m'$ if necessary,
we may assume that $m' = (m' _1, \ldots , m' _l)$ satisfies
\begin{align*}
  0 < m' _1 \leq \dots \leq m' _l \leq 12, & &
  \frac{1}{m' _1} + \dots + \frac{1}{m' _{l - 1}} \leq 1.1822 < \frac{1}{m' _1} + \dots + \frac{1}{m' _l}.
\end{align*}
There are only finitely many lists of integers $m'$
that satisfy these inequalities.
Using a computer program in the arXiv version of this paper,
we verified that all of them are good,
due to either the
sufficient condition in Lemma~\ref{lemma: SumInv sufficient}
or the exhaustive searching mentioned at the end of Section~\ref{subsection: discretized point patrolling}.
This is a contradiction.
\end{proof}

The constant in Theorem~\ref{theorem: 1.546} cannot be made smaller than
$\sum _{i = 0} ^\infty 1 / (2 ^i + 1) = 1.264\ldots {}$:

\begin{theorem}
\label{theorem: bad example}
$(2, 3, 5, \ldots , 2^k+1)$ is bad.
\end{theorem}

\begin{proof}
We prove a stronger claim by induction of $k$:
there is no schedule for
these agents (with minimum gaps $2$, $3$, $5$, \ldots , $2^k + 1$)
that visits all of the $2 ^{k + 1}$ consecutive integers,
say $1$, \ldots, $2^{k+1}$.

When $k = 0$, this is trivial.
Suppose that the claim holds for $k = t - 1$.
Assume for the sake of contradiction that there is
a schedule for $k = t$
(i.e., for the $t + 1$ agents with minimum gaps $2$, $3$, $5$, \ldots , $2 ^t + 1$)
that visits $1$, \ldots , $2 ^{t + 1}$.
Let $p \in \{1, \ldots , 2 ^{t + 1}\}$ be the first integer
that is visited by the agent with minimum gap $2 ^t + 1$.
Then
\begin{itemize}
\item
  if $p > 2 ^t$,
  the other $t$ agents (with minimum gaps $2$, $3$, $5$, \ldots , $2 ^{t -1} + 1$)
  must visit integers $1$, \ldots, $2 ^t$;
\item
  otherwise,
  these $t$ agents
  must visit integers $p + 1$, \ldots , $p + 2 ^t$.
\end{itemize}
In both cases, the induction hypothesis implies that there is no such schedule.
\end{proof}

\begin{conjecture}
If $m$ is a list of positive integers satisfying $\SumInv (m) > \alpha := \sum _{i = 0} ^\infty 1/(2^i+1) \approx 1.264$, then $m$ is good.
\end{conjecture}

\section{Complexity of problems related to point patrolling}
\label{section: complexity}

We have discussed approximation algorithms for patrolling problems.
Finding exactly optimal solutions does not seem to be easy,
and in particular, we conjecture that the
Discretized Point Patrolling Problem is not pseudo-polynomial-time decidable
(i.e., is not polynomial-time decidable even when the input integers $m _1$, \ldots, $m _k$ are written in unary).
However, we have not been able to show any non-trivial hardness result for this problem.
In this section, we prove $\classNP$-completeness of two related problems.

\subsection{Disjoint residue classes}

Consider the special case of the Discretized Point Patrolling Problem
where the input $m$ satisfies $\SumInv (m) = 1$
(see Section~\ref{subsection: suminv} for $\SumInv$).
If $m$ is good in this case,
a periodic schedule for it (which exists, as shown at the end of
Section~\ref{subsection: discretized point patrolling})
must use each agent~$i$ exactly $1 / m _i$ fraction of the time,
and hence, the set of times at which $i$ makes a visit must be
of the form
$m _i \Zset + r _i = \{\, m _i n + r _i : n \in \Zset \,\}$
for some $r _i \in \Zset$.

Thus, this special case
is equivalent to the following problem.
For positive integers $m _1$, \ldots, $m _k$ and
integers $r _1$, \ldots, $r _k$,
the set $\{(m _1, r _1), \ldots , (m _k, r _k)\}$
is called a \emph{disjoint covering system}~\cite{covering} if
the $k$ sets $m _1 \Zset + r _1$, \ldots, $m _k \Zset + r _k$
are a partition of $\Zset$.

\begin{DisjointCoveringSystemProblem}
We are given a list $(m _1, \dots, m _k)$ of positive integers.
Determine whether there are integers $r _1$, \ldots, $r _k$
such that $\{(m _1, r _1), \ldots , (m _k, r _k)\}$ is a disjoint covering system.
\end{DisjointCoveringSystemProblem}

Since it is easy to tell whether a given $m = (m _1, \dots, m _k)$ satisfies
$\SumInv (m) = 1$,
this problem is equivalent to
the above special case of the Discretized Point Patrolling Problem.

By the Chinese remainder theorem,
the sets $m _i \Zset + r _i$ and $m _j \Zset + r _j$ are disjoint if and only if
\begin{equation}
\label{equation: disjointness}
  r _i \not \equiv r _j \pmod{\gcd (m _i, m _j)},
\end{equation}
where $\gcd (m _i, m _j)$ is
the greatest common divisor of $m _i$ and $m _j$.
Since this condition~\eqref{equation: disjointness}
can be checked easily (for all pairs $(i, j)$),
the Disjoint Covering System Problem belongs to $\classNP$.
We conjecture that
it is $\classNP$-complete:

\begin{conjecture}
The Disjoint Covering System Problem is $\classNP$-complete,
even if the inputs $m _1$, \ldots, $m _k$ are written in unary.
\end{conjecture}

Failing to prove this, we consider $\classNP$-completeness of
a similar problem.
For positive integers $m _1$, \ldots, $m _k$ and
integers $r _1$, \ldots, $r _k$,
the set $\{(m _1, r _1), \ldots, (m _k, r _k)\}$
is called a \emph{disjoint residue class}~\cite{residue} if
the $k$ sets $m _1 \Zset + r _1$, \ldots, $m _k \Zset + r _k$ are
pairwise disjoint.
Thus, this time the sets need not cover all integers.

\begin{DisjointResidueClassProblem}
We are given a list $(m _1, \dots, m _k)$ of positive integers.
Determine whether there are integers $r _1$, \ldots, $r _k$
such that $\{(m _1, r _1), \ldots, (m _k, r _k)\}$ is a disjoint residue class.
\end{DisjointResidueClassProblem}

\begin{theorem}
The Disjoint Residue Class Problem is $\classNP$-complete,
even if the inputs $m _1$, \ldots, $m _k$ are written in unary.
\end{theorem}

\begin{proof}
The vertex cover problem for triangle-free graphs is known to be $\classNP$-complete.
We reduce this problem to the Disjoint Residue Class Problem.
Let $G = (V, E)$ be a triangle-free graph
with $n = \lvert V \rvert $ vertices
and $k = \lvert E \rvert $ edges,
and let $s \leq n$ be a positive integer.
To the $n$ vertices $u \in V$ of $G$,
we assign distinct prime numbers $p _u$, all greater than $n$.
We rename the edges so that $E = \{1, \ldots, k\}$,
and for each edge $i \in E$ with endpoints $u$, $v \in V$,
we let $m _i = s p _u p _v$.
We claim that $G$ has a vertex cover of size $s$ if and only if
$(m _i) _{i \in E}$ is a yes-instance of the Disjoint Residue Class Problem.

Suppose that $\{v _1, \ldots, v _s\} \subseteq V$ is a vertex cover of size $s$.
That is, there is a mapping $a \colon E \to \{1, \ldots, s\}$
such that each edge $i \in E$ is incident to $v _{a (i)}$.
Since less than $n$ edges $i$ are mapped to the same $a (i)$,
there is a mapping $b \colon E \to \{1, \ldots , n\}$ such that
$b (i) \neq b (j)$ for all distinct $i$, $j \in E$ with $a (i) = b (i)$.
For each $i \in E$,
let $r _i$ be such that $r _i \equiv a (i) \pmod s$ and $r _i \equiv b (i) \pmod{p _{v _{a (i)}}}$.
Then $\{\, (m _i, r _i) : i \in E \,\}$ is a disjoint residue class,
because
the condition~\eqref{equation: disjointness}
holds for all distinct $i$, $j \in E$ as follows:
\begin{itemize}
\item If $a (i) = a (j)$ (and hence $b (i) \neq b (j)$),
  then $p _{v _{a (i)}}$ is a common divisor of $m _i$ and $m _j$,
  and $r _i \equiv b (i) \not \equiv b (j) \equiv r _j \pmod{p _{v _{a (i)}}}$.
\item Otherwise,
  $s$ is a common divisor of $m _i$ and $m _j$,
  and $r _i \equiv a (i) \not \equiv a (j) \equiv r _j \pmod s$.
\end{itemize}

Conversely, let $\{\, (m _i, r _i) : i \in E \,\}$ be a disjoint residue class.
For each $a = 1$, \ldots, $s$,
let $E _a = \{\, i \in E : r _i \equiv a \pmod s \,\}$.
Then every pair of edges $i$, $j$ in $E _a$ must share a vertex,
because otherwise
$\gcd (m _i, m _j) = s$ and
the condition~\eqref{equation: disjointness} is violated.
Since $G$ is triangle-free,
there must exist a vertex $v _a$ shared by all edges in $E _a$.
Thus, $\{v _1, \ldots , v _s\}$ is a vertex cover.
\end{proof}

\subsection{Patrolling a specified set of times}

We also obtain an $\classNP$-complete problem if we specify the set of times at which the point must be visited:

\begin{GeneralizedPointPatrollingProblem}
We are given a finite set of times $S \subseteq \Zset$
and a list $(m _1, \ldots, m _k)$ of positive integers.
The gap between two distinct visits by agent~$i$
must be at least $m _i$.
Determine whether there is a schedule of visits in which,
at each time in $S$, at least one agent makes a visit.
\end{GeneralizedPointPatrollingProblem}

\begin{theorem}
\label{theorem: generalized point patrolling np-complete}
The Generalized Point Patrolling Problem is $\classNP$-complete,
even if the inputs $m _1$, \ldots, $m _k$ are written in unary.
\end{theorem}

\begin{proof}
In \emph{Numerical 3-Dimensional Matching},
we are given three sequences
$(x _1, \ldots, x _k)$, $(y _1, \ldots, y _k)$, $(z _1, \ldots, z _k)$
of positive integers written in unary,
and we are asked whether there are
permutations $(p _1, \ldots, p _k)$ and $(q _1, \ldots, q _k)$
of $(1, \ldots, k)$
such that
\begin{align}
\label{equation: 3DM yes}
  x _{p _i} + y _{q _i} + z _i & = b, & i & = 1, \ldots , k,
\end{align}
where $b = (x _1 + \dots + x _k + y _1 + \dots + y _k + z _1 + \dots + z _k) / k$.
We reduce this problem, which is known to be
$\classNP$-complete~\cite{book},
to the Generalized Point Patrolling Problem as follows:
given the above instance,
we construct the input  $(S, m _1, \ldots , m _k)$ by
choosing a sufficiently large number $M$
(say, $M = 3 k (b + \max \{x _1, \ldots, x _k, y _1, \ldots, y _k\})$)
and setting 
\begin{align*}
\label{equation: m_i for 3DM}
X _i & = 3 k x _i - i, \qquad Y _i = M - 3 k y _i + i, \qquad
m _i = M - 3 k (b - z _i), & i & = 1, \ldots, k, 
\end{align*}
and $S = \{X _1, \ldots, X _k, Y _1, \ldots, Y _k\}$. 
This can clearly be done in polynomial time (even when we need to write $m _i$ in unary).
Note that for each triple $(p, q, i) \in \{1, \ldots, k\} ^3$, 
we have $x _p + y _q + z _i \leq b$ if and only if 
$Y _q - X _p \geq m _i$. 

Thus, if we have \eqref{equation: 3DM yes}, 
there is a patrolling schedule for $(S, m _1, \ldots, m _k)$ where
each agent~$i$ makes visits at $X_{p_{i}}$ and $Y_{q_{i}}$.

Conversely, suppose that there is
a patrolling schedule for $(S, m _1, \ldots, m _k)$.
Since $M$ is (and hence $m _i$ are) sufficiently large, 
no agent can visit more than one of ${X_1, \ldots, X_k}$ or more than one of ${Y_1, \ldots, Y_k}$. 
Since the $k$ agents visit the $2 k$ times in $S$,
each agent~$i$ visits one of ${X_1, \ldots, X_k}$ and one of ${Y_1, \ldots, Y_k}$, 
say $X _{p _{i}}$ and $Y _{q _{i}}$.
Since $Y _{q _i} - X _{p _i} \geq m _i$,
we have $x_{p_i} + y_{q_i} + z_i \leq b$. 
Since this holds for all~$i$, we have \eqref{equation: 3DM yes}.
\end{proof}

\end{document}